\documentclass{article}

\newtheorem{theorem}{Theorem}

\newtheorem{remark}[theorem]{Remark}

\newenvironment{proof}[1][Proof]{\noindent\textbf{#1.} }{\ \rule{0.5em}{0.5em}}
\input{tcilatex}
\input{graphicx}

\begin{document}

\title{Establishing Cryptocurrency Equilibria Through Game Theory}

\author{Carey Caginalp$^{a,b}$ and Gunduz Caginalp$^{a}$}
\date{$^{a}$University of Pittsburgh, Mathematics Department, 301 Thackeray Hall, Pittsburgh, PA, USA\\
$^{b}$Chapman University, Economic Science Institute, 1 University Drive, Orange, CA, USA\\
\today}
\affiliation{University of Pittsburgh, Mathematics Department, 301 Thackeray Hall, Pittsburgh, PA, USA}
\affiliation{Chapman University, Economic Science Institute, 1 University Drive, Orange, CA, USA}
\maketitle

\begin{abstract}We utilize optimization methods to determine equilibria of
cryptocurrencies. A core group, the wealthy, fears the loss of assets that can be seized
by a government. Volatility may be influenced by speculators. The wealthy must divide their
assets between the home currency and the cryptocurrency, while the
government decides the probability of seizing a fraction the assets of
this group. We establish conditions for existence and uniqueness of
Nash equilibria. Also examined is the separate timescale problem
in which the government policy cannot be reversed, while the wealthy can
adjust their allocation in reaction to the government's designation of
probability.
\end{abstract}

Cryptocurrencies have evolved into a new speculative asset form that differs
from others in that most represent no intrinsic value; they cannot be redeemed
by a financial institution for any amount \cite{BER17}. The roller-coaster
ride of Bitcoin prices\footnote{In March 2019, Bitcoin hovered near \$3,800.}
from \$6,000 to \$20,000 back to \$6,000, with bounces in between, all during
the period from October 2017 to July 2018, was shadowed by other major
cryptocurrencies \cite{CUR17}. This has been accompanied by the general
feeling in government, business and academia that the speculative fever is of
concern only to those who own the cryptocurrencies. There is some
justification for this perspective as the total market capitalization of all
cryptocurrencies is now only about \$131 billion, so that large moves in the
cryptocurrency price are not likely to have a significant impact on the
world's stock and bond markets. However, this impact will present a
significant risk to the world's markets if the market capitalization of the
cryptocurrencies increases significantly. During the dramatic round trip of
Bitcoin between \$6,000 to \$20,000, the market capitalization of all
cryptocurrencies nearly doubled in six months. Moreover, 10,000 Bitcoins were
used to purchase two pizzas in the first transaction in 2010 \cite{PRI17}. If
people gradually become more comfortable with cryptocurrencies, as they did
with internet shopping, it is likely that the market capitalization could grow
to a few percent of the \$75 trillion Gross World Product (GWP) as the
fraction of the world's savings that is under threat by government seizure,
high inflation, etc., is certainly at least this fraction (as discussed
further below). At this point, large price changes in cryptocurrencies would
likely have an impact on the broader markets.

It is thus important to understand the factors behind the market
capitalization and price of cryptocurrencies. With all other assets there is
some theoretical methodology to determine the value, which is at least a first
step estimating the trading price. For example, the value of a stock is
assessed by measures such as the expected dividend stream (see Graham
\cite{g}, \cite{gm}, Luenberger \cite{LUE97}, Bodie et al. \cite{bkm}, and
Wolpert et. al. \cite{JOS12}). Even beyond these calculations, a shareholder
is \textit{de facto} part owner of a corporation, and shareholders can -- and
do -- collectively exercise their rights assured by law. By contrast, a
typical cryptocurrency does not assure the owner of any rights. Furthermore,
there is no corporate governance at all. The "miners and developers" -- whose
names are usually not disclosed -- get together and decide essentially on the
supply (e.g., by introducing a related cryptocurrency that they term a
"fork"). Unlike corporate actions in which shareholders can demand a vote,
e.g., for directors via a proxy battle, it is not even clear which, if any,
nation's laws apply. Thus, the absence of an intrinsic value of a
cryptocurrency means that the usual traditional finance methods, such as those
introduced by Graham (\cite{g}, \cite{gm}), are inapplicable.

Our analysis begins with a game theoretic examination of the motivations of
three groups that are the key players. For a core group, the basic need for a
cryptocurrency arises from the inadequacy of the home currency and banking
system \cite{RAP18}. There are also a significant number of people who are not
able to obtain a credit card or even open a bank account in the US, for
example \cite{WHI17}. In many countries, owning large amounts of the currency
can present a significant risk. There is the possibility of expropriation by
the government, sometimes in the guise of a corruption probe. The government
could institute policies in which inflation is very high, e.g., the extreme
example in Venezuela \cite{GIL17} where hyperinflation decimated any
individual savings. Onerous taxes can be placed by the government on the
wealthy. Thus a group of people in the world have rational reasons to replace
their country's currency with one that is outside the control of their
government or financial institutions, even if it presents some risk. Once it is
transferred to cryptocurrency, they would have the option of buying a more
reliable currency or asset in another country. We denote this group by
$\mathcal{W}$ (the "wealthy"). Returning to the point made above, there is a
substantial amount of the world's wealth (including individual's whose assets
are not large) that is in this situation. Many of these people, however, are
not yet comfortable with or knowledgeable about cryptocurrencies. As they feel
more comfortable, a greater fraction of this wealth may move into
cryptocurrencies, inflating the market capitalization, perhaps to a few
percent of the world's GWP.

The second group, $\mathcal{D}$, represents a government that is totalitarian,
at least with respect to monetary policy, so that its citizens are not free to
transfer their wealth into other, more reliable national currencies. There is
a probability, $p,$ that the government can initiate policies that will
deprive citizens of a fraction $k$ of their wealth, e.g., by printing money.
This possibility is noted by the wealthy, $\mathcal{W}$, who must make a
decision on the fraction of their assets, denoted $1-x$ held in the home
currency and the remainder, $x$, in cryptocurrency, which presents risks of
its own due to the volatility. The government, $\mathcal{D}$, exhibits risk
aversion as with any financial entity. After all, its existence is dependent
on obtaining funds from its citizens. A third group, $\mathcal{S}$, consists
of the speculators\footnote{In commodities such as gold and oil, there are
producers and industrial users who must trade. In cryptocurrencies, there are
no end users other than $\mathcal{W}$ \ who are trade infrequently, unlike
industrial users for gold, for example, who are perpetually trading.} whose
sole reason to trade is to profit from the transaction at the expense of the
less knowledgeable group, $\mathcal{W}$. In a typical situation, the a member
of $\mathcal{W}$ is trading for the first or second time -- having made their
money in another endeavor -- while the speculators are professionals who have
made thousands of trades, and make their living at the expense of novice
traders. The speculators effectively determine the volatility (see Appendix).
Note, however, that our analysis would be similar if the volatility were an
exogenous variable that is set arbitrarily. While $\mathcal{D}$ can set the
probability, $p$, with which the assets can be seized, group $\mathcal{W}$ can
decide what fraction, $x,$ to convert into the cryptocurrency.

We model this situation to find equilibria in two different ways. The first is
to find the Nash equilibrium \cite{m}, \cite{t}, \cite{MIL90}, \cite{SMI05},
\cite{ca} which is defined as the point $\left(  p^{\ast},x^{\ast}\right)  $
such that neither party can improve its fortunes by unilateral action. The
underlying assumption is that both parties, $\mathcal{W}$ and $\mathcal{D}$,
are aware of the situation faced by the other, so that they can simultaneously
self-optimize while assuming that the other party does likewise.

In a later section, we utilize the more realistic assumption that while
$\mathcal{W}$ can make immediate changes (e.g., one day), $\mathcal{D}$ must
make a decision that is irrevocable during a longer time (e.g., one year) as
policies (e.g., creating inflation, imposing onerous taxes) are implemented.
But in doing so, $\mathcal{D}$ must be aware that $\mathcal{W}$ will
self-optimize in its choice of $x$, knowing $p$. Thus, both parties are aware
of the different time scales involved in anticipating the other party's decision.

The methods we present in this paper are aimed at determining the demand using
optimization. The investors of the cryptocurrency do not have any clear idea
how much is the right amount to pay per unit of the cryptocurrency, so that
the demand will determine the trading price as discussed in Appendix A and
\ \cite{cc}. In a setting in which there is one cryptocurrency with a fixed
supply, the price will be determined as%
\begin{equation}
\text{Equilibrium Trading Price}=\frac{\text{Demand in Dollars}}{\text{Number
of Units}}. \label{demand}%
\end{equation}
Analogous methodology can be utilized for multiple cryptocurrencies.

\bigskip

One might be led to examining cryptocurrencies in the context of monetary
policy, but the terminology "currency" is the main similarity between the
cryptocurrencies and the US Dollar, Euro, Yen, etc. The differences are
profound\footnote{This point is probably clear to anyone who bought Bitcoin at
$\$20,000$ and sold it at $\$6,000$ several months later.}. Major currencies
are established by governments within a well-defined process that is governed
by law. The identities of those who are responsible for monetary policy are
known. The citizens of the country can influence the representatives who
appoint the monetary officials. Finally, if the citizens feel that the
direction of monetary policy is not in their interest, they can elect new
representatives. There have been many currencies throughout the world, and it
is no accident that the most viable currencies have been those of the
countries with the most respect for the law and the voices of their citizens.

Thus, the theory of monetary policy will be of limited use in the
understanding today's cryptocurrencies. The aspect of our analysis that is
closest to monetary policy involves the actions of the government,
$\mathcal{D}$, which must make a decision on issues such as generating
inflation (see, e.g., \cite{B}, \cite{BG}). Of course, the government in our
analysis is one that is very different from the major democratic governments
that have the more reliable currencies.

Since the widespread use of cryptocurrencies is a fairly new phenomenon, the
literature is also recent. Many papers have focussed on the blockchain
technology and its potential for increased speed and safety of transactions.
The introduction of JP\ Morgan's JPM Coin (see Appendix A) is an example of
utilization of this technology without any new economic issues, since
JPM\ Coin would be redeemable in US\ Dollars. The economics ofcryptocurrencies
have been discussed in terms of legal issues \cite{JV}, valuation
\cite{CR}, security issues \cite{BR}, \cite{CU} and stability \cite{cac},
\cite{IKMS}, \cite{LT} , \cite{OCK} and feasibility \cite{CG}. \ Experiments
have also been used to study cryptocurrencies and related issues \cite{CM},
\cite{GW}.

Our analysis can be viewed in the more general setting of an asset that is
easily traded and out of the reach of the state and other entities. However,
the popularity of cryptocurrencies may indicate that there are not so many of
these. Traditionally, gold has been used as a haven, but it is not always easy
to prevent theft. Nevertheless, followers have often noted spikes in gold when
there is political uncertainty in the highly populated and less developed
countries. Also, the demand for gold depends upon other factors such as
industrial use.
\section*{The Utility Functions of the Groups}

The general framework for
this section will be to write the utility functions of the three groups,
modeled on portfolio theory \cite{LUE97}, \cite{bkm} whereby one seeks allocate resources
to maximize return while minimizing risk. The general form for a basic
utility function is $U=m-d^{2}\sigma ^{2}$ where $m$ and $\sigma ^{2}$ are
the expectation and variance of the outcome, while the parameter $d^{2}$
quantifies the risk aversion.

The speculators,\ $\mathcal{S}$, are assumed to have an influence on the 
volatility and risk. Even if they had no influence on the volatility, they are 
likely to profit at the expense of $\mathcal{W}$, who are
likely to be novices. Hence, the role of $\mathcal{S}$ is secondary (and
discussed in the Appendix) as they create an expected loss and a variance
for $\mathcal{W}$.

Focusing now on the groups, $\mathcal{W}$ and $\mathcal{D}$, we assume that $%
\mathcal{W}$ has a choice between the home currency,~F, and a
cryptocurrency, Y. Any money held in F faces the risk that a fixed fraction $%
k\in \left[ 0,1\right] $ will be seized by $\mathcal{D}$ with probability $p.
$ Thus, the outcome will be $\left( 1-x\right) \left( 1-k\right) $ with
probability $p$ and $1$ with probability $1-p.$ Letting $m_{F}$ and $\sigma
_{F}^{2}$ denote the mean and variance of the investment in F, one finds,

\begin{eqnarray}
m_{F} &=&\left( 1-k\right) p+1\cdot \left( 1-p\right) =1-kp,  \nonumber \\
\sigma _{F}^{2} &=&k^{2}p\left( 1-p\right) .  \label{meanF}
\end{eqnarray}%
For the investment in Y, we let $m_{Y}$ and $\sigma _{Y}^{2}$ denote the
mean and variance that will be determined by the speculators (see Appendix).

The utility function for $\mathcal{W}$ with the fraction $x\in \left[ 0,1%
\right] $ of its assets in Y and the remainder in F can then be expressed as%
\begin{eqnarray}
U_{W} &=&m-d^{2}\sigma ^{2}  \label{uw} \\
&=&xm_{Y}+\left( 1-x\right) m_{F}\nonumber\\
&-&d^{2}\left\{ x^{2}\sigma _{Y}^{2}+\left(
1-x\right) ^{2}\sigma _{F}^{2}+2x\left( 1-x\right) Cov\left[ Y,F\right]
\right\} .  \nonumber
\end{eqnarray}%
We will assume that the correlation between the two assets, $Y$ and $F,$ is
zero, but the analysis can easily be carried out if there is a correlation.

The utility function for $\mathcal{D}$ can be expressed in terms of the amount that
it seizes, i.e.,%
\begin{equation}
U_{D}=\left( 1-x\right) kp.  \label{ud}
\end{equation}%
This can be augmented with a term (as in portfolio theory) that expresses
the risk aversion. In particular, one has
\begin{eqnarray}
U_{D} = \left( 1-x\right) kp-d_{D}^{2}p^{2}, \label{udvar}
\end{eqnarray}%
where $d_{D}^{2}$ represents the risk aversion of $\mathcal{D}$.
\section*{Nash Equilibria}

We assume the utility functions described in
Section 3, using the risk aversion form of $U_{D}$ above $\left( \ref{udvar}%
\right) .$ Thus, we need to find $\left( p^{\ast },x^{\ast }\right) $ such
that 
\begin{eqnarray*}
\partial _{x}U_{W}\left( p^{\ast },x^{\ast }\right) &=&0,\ \ \partial
_{p}U_{D}\left( p^{\ast },x^{\ast }\right) =0, \\
\partial _{xx}U_{W}\left( p^{\ast },x^{\ast }\right) &\leq &0,\ \ \partial
_{pp}U_{D}\left( p^{\ast },x^{\ast }\right) \leq 0\ .
\end{eqnarray*}%
i.e.,$\ \left( p^{\ast },x^{\ast }\right) $ satisfies the definition of a
Nash equilibrium (see e.g., \cite{m},). Briefly, the definition ensures that
at$\ \left( p^{\ast },x^{\ast }\right) $ neither party can unilaterally
improve its situation. We compute%
\begin{eqnarray}
0 &=&\partial _{p}U_{D}\left( p,x\right) =\left( 1-x\right) k-2d_{D}^{2}p,
\label{pder} \\
0 &=&\partial _{x}U_{W}\left( p,x\right) =m_{Y}-m_{F}+2d^{2}\sigma
_{F}^{2}-2d^{2}\left( \sigma _{Y}^{2}+\sigma _{F}^{2}\right) x  \nonumber \\
&=&m_{Y}-1+kp+2d^{2}k^{2}p\left( 1-p\right) -2d^{2}\left( \sigma
_{Y}^{2}+k^{2}p\left( 1-p\right) \right) x.  \label{xder}
\end{eqnarray}%
Denote the solution of $\left( \ref{pder}\right) $ by $x_{1}\left( p\right) $
and that of $\left( \ref{xder}\right) $ by $x_{2}\left( p\right) ,$ so that 
\begin{equation}
x_{1}\left( p\right) =1-\frac{2d_{D}^{2}p}{k}  \label{x1}
\end{equation}

\begin{eqnarray}
x_{2}\left( p\right) &=&\frac{\sigma _{F}^{2}+\left( 2d^{2}\right)
^{-1}\left( m_{Y}-m_{F}\right) }{\sigma _{F}^{2}+\sigma _{Y}^{2}}  \nonumber
\\
&=&\frac{k^{2}p\left( 1-p\right) +\left( 2d^{2}\right) ^{-1}\left(
m_{Y}-1+kp\right) }{k^{2}p\left( 1-p\right) +\sigma _{Y}^{2}}.  \label{xtwo}
\end{eqnarray}%
The intersection of $x_{1}\left( p\right) $ and $x_{2}\left( p\right) $
determine a Nash equilibrium. We first establish sufficient conditions for
at most one equilibrium, and then prove that under broad conditions, there
exists a Nash equilibrium. Some of these curves for sample values of the 
parameters are illustrated in Figure 2 (see published version for figures). 


\begin{theorem}
For $p\in \left[ 0,1/2\right] $ one has $x_{2}^{\prime
}\left( p\right) \geq 0$ for all values of the parameters, so there can be
at most one value of $p$ such that $x_{1}\left( p\right) =x_{2}\left(
p\right) $, and thus at most one Nash equilibrium for $p\in \left[ 0,1/2%
\right] $.
\label{3.1}
\end{theorem}

\begin{proof}
For convenience set $f\left( p\right) =k^{2}p\left( 1-p\right) ,$ $%
c_{1}=\left( 2d^{2}\right) ^{-1}\left( 1-m_{Y}\right) ,$ $c_{2}=\left(
2d^{2}\right) ^{-1}k$ and $c_{3}=\sigma _{Y}^{2}$ so that%
\[
x_{2}\left( p\right) =\frac{f\left( p\right) +c_{2}p-c_{1}}{f\left( p\right)
+c_{3}} 
\]%
and 
\begin{equation}
x_{2}^{\prime }\left( p\right) =\frac{c_{2}k^{2}p^{2}+c_{2}c_{3}+\left(
c_{1}+c_{3}\right) k^{2}\left( 1-2p\right) }{\left[ f\left( p\right) +c_{3}%
\right] ^{2}}.  \label{xprime}
\end{equation}%
Clearly, for $p\in \left[ 0,1/2\right] $ all terms are positive and the
conclusion follows.
\end{proof}

\begin{theorem}
If the parameters $d,$ $k$ and $m_{Y}$ satisfy 
\begin{equation}
c_{1}+c_{3}\leq c_{2}\text{ i.e., }\left( 1+2d^{2}\right) \left(
1-m_{Y}\right) \leq k  \label{condpos}
\end{equation}%
then $x_{2}^{\prime }\left( p\right) \geq 0$ for all $p\in \left[ 0,1\right]
.$ Thus there can be at most one Nash equilibrium under these conditions.
\label{3.2}
\end{theorem}

\begin{proof}
For $p\in \left[ 0,1/2\right] $ the result has been established. For $%
p\in \left[ 1/2,1\right] $ the numerator of $\left( \ref{xprime}\right) $
is%
\[
c_{2}k^{2}p^{2}+c_{2}c_{3}+c_{2}k^{2}\left( 1-2p\right)
=c_{2}c_{3}+c_{2}k^{2}\left( 1-p\right) ^{2}>0, 
\]%
and the result follows.
\end{proof}

\bigskip

Having determined sufficient conditions for uniqueness, we now focus on
establishing existence of Nash equilibrium. Note first that the $p-$%
intercept of $x_{1}\left( p\right) $ can be on either side of $x=1$
depending on the slope $-2d_{D}^{2}/k.$ In particular, we let $%
p_{c}:=k\left( 2d_{D}^{2}\right) ^{-1}$ and consider the two cases
separately.

\bigskip

\begin{theorem}
$\left( a\right) $ If $p_{c}:=k\left( 2d_{D}^{2}\right) ^{-1}<1$
and 
\begin{equation}
k^{2}p_{c}\left( 1-p_{c}\right) +\left( 2d^{2}\right) ^{-1}\left(
m_{Y}-1+p_{c}k\right) >0,  \label{condA}
\end{equation}%
then one has a Nash equilibrium, i.e., there exists $\left( p^{\ast
},x^{\ast }\right) \in \left[ 0,1\right] \times \left[ 0,1\right] $ such
that $x_{1}\left( p^{\ast }\right) =x_{2}\left( p^{\ast }\right) =x^{\ast }.$

$\left( b\right) $ If $p_{c}:=k\left( 2d_{D}^{2}\right) ^{-1}\geq 1$ and 
\[
\frac{\left( 2d^{2}\right) ^{-1}\left( m_{Y}-1+k\right) }{\sigma _{Y}^{2}}+%
\frac{2d_{D}^{2}}{k}\geq 1 
\]%
then one has again a Nash equilibrium.

If in addition, equation $\left( \ref{condpos}\right) $ holds, then the Nash
equilibrium $\left( p^{\ast },x^{\ast }\right) $ is unique.
\label{3.3}
\end{theorem}

\begin{proof}
Recall that $m_{Y}\leq 1$. Thus, we have the inequality, 
\[
1=x_{1}\left( 0\right) >0\geq x_{2}\left( 0\right) =\left( m_{Y}-1\right)
/\left( 2d^{2}\sigma _{Y}^{2}\right) . 
\]%
We use the Intermediate Value Theorem to establish an intersection between $%
x_{1}\left( p\right) $ and $x_{2}\left( p\right) $ in the unit square in $%
\left( p,x\right) $ space.

Case $\left( a\right) .$ Suppose $p_{c}:=k\left( 2d_{D}^{2}\right) ^{-1}<1.$
Then 
\[
x_{2}\left( p_{c}\right) =\frac{k^{2}p_{c}\left( 1-p_{c}\right) +\left(
2d^{2}\right) ^{-1}\left( m_{Y}-1+p_{c}k\right) }{k^{2}p_{c}\left(
1-p_{c}\right) +\sigma _{Y}^{2}} 
\]
so that by $\left( \ref{condA}\right) $ one has $x_{2}\left( p_{c}\right)
\geq 0=x_{1}\left( p_{c}\right) $. Thus there exists an intersection of $%
x_{1}\left( p\right) $ and $x_{2}\left( p\right) $ at $\left( p^{\ast
},x^{\ast }\right) \in \left[ 0,0\right] \times \left[ p_{c},x\left(
p_{c}\right) \right] \in \left[ 0,1\right] \times \left[ 0,1\right] .$

Case $\left( b\right) .$ Suppose $p_{c}:=k\left( 2d_{D}^{2}\right) ^{-1}>1.$
Then $x_{1}\left( 1\right) >0,$ and $x_{1}\left( p\right) \in \left[ 0,1%
\right] $ for $p\in \left[ 0,1\right] $. Thus, an intersection of $%
x_{1}\left( p\right) $ and $x_{2}\left( p\right) $ for $p\in \left[ 0,1%
\right] $ must occur on the unit square provided that $x_{2}\left( 1\right)
\geq x_{1}\left( 1\right) .$ Then the required condition is 
\begin{eqnarray*}
x_{2}\left( 1\right) &=&\frac{\left( 2d^{2}\right) ^{-1}\left(
m_{Y}-1+k\right) }{\sigma _{Y}^{2}} \\
&\geq &1-\frac{2d_{D}^{2}}{k}=x_{1}\left( 1\right) .
\end{eqnarray*}
\end{proof}

\begin{remark}
The Nash equilibrium may not be unique if the condition above, i.e.,%
\[
\left( 1+2d^{2}\right) \left( 1-m_{Y}\right) \leq k 
\]%
of Theorem \ref{3.3} is violated. An example for two Nash equilibria can be
constructed with the parameters:%
\[
k=0.7,\ m_{Y}=0.8,\ d=2,\ d_{D}=0.355,\ \sigma _{Y}^{2}=0.1. 
\]%
The two equilibria are given approximately by $\left(p^{*},x^{*}\right)
=\left(0.88,0.68\right)$ and $\left(0.96,0.65\right)$, as pictured in 
Figure 1. 

\end{remark}

\section*{Equilibrium with Disparate Time Scales}

We consider the
situation in which the wealthy, $\mathcal{W}$, can decide on an allocation $x
$ immediately, (e.g., within one day), and adjust to the probability, $p,$
while $\mathcal{D}$ must set $p$ that cannot be changed for a long time
e.g., one year. Thus, $\mathcal{D}$ lacks the opportunity to react to the
value of $x.$ Both parties are aware of the position of the other group.
Hence, $\mathcal{D}$ knows that once he sets $p,$ group $\mathcal{W}$ will
set $x=\hat{x}\left( p\right) $ in a way that optimizes $U_{W}$, and that $%
\mathcal{W}$ does not need to be concerned with any readjustment of $p$ in
reaction to their choice of $x.$ Thus, $\mathcal{D}$ must examine $U_{W}$
(based on the publicly available information on the volatility of Y) and
decide on a value of $p$ that will optimize $U_{D}\left(p, \hat{x}\left(
p\right)\right) .$ Within this setting the utility of $\mathcal{D}$ need
not be strictly convex in order for an interior maximum (i.e. such that $%
0<p<1$ and $0<x<1$). Thus, we consider the case in which $\mathcal{D}$ has
utility that is proportional to the amount it takes, without any risk
aversion, which can be included with a bit more calculation.

We define the quantity%
\begin{equation}
A:=2d^{2}\sigma _{Y}^{2}+1-m_{Y}  \label{Adef}
\end{equation}%
which arises naturally in the calculations and is a measure of the risk and
expected loss from $Y.$ Thus a higher value of $A$ means $Y$ is less
attractive to the wealthy.

\begin{theorem}
Suppose that the utility functions, $U_{W}$ and $U_{D},
$ given by%
\[
U_{D}=\left( 1-x\right) kp\ \ \ \ \ \ \ \ \ \ \ \ \ \ \ \ \ \ \ \ \ \ \ \ \
\ \ \ \ \ \ \ \ \ \ \ \ \ \ \ \ \ \ \ \ \ \ \ \ \ \ \ \ \ \ \ \ \ \ \ \ \ \
\ \ \ \ \ \ \ \ \ \ \ \ \ \ \ \ \ \ \ \ \ \ \ \ \ \ \ 
\]%
\begin{eqnarray*}
U_{W} &=&m-d^{2}\sigma ^{2} \\
&=&xm_{Y}+\left( 1-x\right) m_{F}\nonumber\\
&-&d^{2}\left\{ x^{2}\sigma _{Y}^{2}+\left(
1-x\right) ^{2}\sigma _{F}^{2}+2x\left( 1-x\right) Cov\left[ Y,F\right]
\right\} .
\end{eqnarray*}%
are known to both parties. Assume that $\mathcal{D}$ sets $p$ irrevocably to
maximize $U_{D}$, while $\mathcal{W}$ chooses $x$ to maximize $U_{W}$ based
on a knowledge of $p.$ For $0\leq A<k$ the optimal choice of $x$ given $p$ is%
\begin{equation}
\hat{x}\left( p\right) :=\frac{m_{Y}-m_{F}}{2d^{2}\left( \sigma
_{Y}^{2}+\sigma _{F}^{2}\right) }+\frac{\sigma _{F}^{2}}{\sigma
_{Y}^{2}+\sigma _{F}^{2}}  \label{xhat1}
\end{equation}%
with $m_{Y}$ and $\sigma _{F}^{2}$ given by $\left( \ref{meanF}\right) ,$
and the optimal value of $p$ is given by 
\begin{equation}
p^{\ast }:=\frac{\sigma _{Y}^{2}}{k\left( k-A\right) }\left( \sqrt{1+\frac{A%
}{\sigma _{Y}^{2}}\left( k-A\right) }-1\right) \ .  \label{pstar}
\end{equation}%
Thus the optimal point is $\left( p,x\right) =\left( p^{\ast },\hat{x}\left(
p^{\ast }\right) \right) .$ The value of maximum, $x^{\ast }=$ $\hat{x}%
\left( p\right) $ is $0$ if the right hand side of $\left( \ref{xhat1}\right) $
is negative, and $1$ if the right hand side exceeds $1$.
\end{theorem}

\begin{remark}
Note that given $p$ the optimal fraction of assets in the
cryptocurrency is a sum of the relative variance of the home currency, i.e., 
$\sigma _{F}^{2}$ as a fraction of $\sigma _{Y}^{2}+\sigma _{F}^{2}$ plus
the difference in expected loss from the home currency, i.e., $1-m_{F}$
minus the expected loss from the cryptocurrency, $1-m_{Y}$ scaled by a risk
aversion factor. Thus the fraction invested in the cryptocurrency increases
as the expected loss and the variance of the home currency increases, 
and conversely.
\end{remark}
\begin{remark}
Note that one obtains an interior maximum with a linear utility function
for $U_{D}$ in this type of optimization, i.e., even though $\mathcal{D}$ 
is interested in pure maximization of its revenue.
\end{remark}

\begin{proof}
Using $\left( \ref{uw}\right) $ we determine the maximum of $U_{W}$
for a fixed $p,$ so that 
\[
0=\partial _{x}U_{W}\left( p,x\right) =m_{Y}-m_{F}+2d^{2}\sigma
_{Y}^{2}-2d^{2}\left( \sigma _{Y}^{2}+\sigma _{F}^{2}\right) x. 
\]%
Noting that $\partial _{xx}U_{W}\left( p,x\right) =-2d^{2}\left( \sigma
_{Y}^{2}+\sigma _{F}^{2}\right) <0$ we see that $U_{D}$ is maximized by $%
\hat{x}\left( p\right) $ given by $\left( \ref{xhat1}\right) $ provided $%
\hat{x}\left( p\right) \in \left[ 0,1\right] $. In the following two cases
the maximum is on the boundary: 
\[
\frac{m_{Y}-m_{F}}{2d^{2}}+\sigma _{F}^{2}<0\text{ \ implies \ }\hat{x}%
\left( p\right) =0, 
\]%
\[
\frac{m_{Y}-m_{F}}{2d^{2}}+\sigma _{F}^{2}>1\text{ \ implies \ }\hat{x}%
\left( p\right) =1. 
\]

Thus, $\hat{x}\left( p\right) $ interpolates between $0$ and $1$ by favoring 
$Y$ if the relative risk of $F$ (measured by $\sigma _{F}^{2}\left( \sigma
_{Y}^{2}+\sigma _{F}^{2}\right) ^{-1}$ is large in comparison with the
relatively greater expected loss in $Y$ (scaled by the sum of the variances).

In anticipation, $\mathcal{D}$ optimizes $U_{D}\left( p,x\left( p\right)
\right) .$ We thus compute, with $B:=\sigma _{Y}^{2}$,%
\begin{eqnarray*}
0 &=&\frac{2d^{2}}{k}\partial _{p}U_{D}\left( p,x\left( p\right) \right)
=\partial _{p}\frac{Ap-kp^{2}}{B+k^{2}p\left( 1-p\right) } \\
&=&\frac{\left( A-2kp\right) \left[ B+k^{2}p\left( 1-p\right) \right]
-\left( Ap-k^{2}p\right) \left[ k^{2}\left( 1-2p\right) \right] }{\left[
B+k^{2}p\left( 1-p\right) \right] ^{2}}.
\end{eqnarray*}%
This identity is equivalent to 
\begin{equation}
Q\left( p\right) :=AB-2Bkp+k^{2}\left( A-k\right) p^{2}=0.  \label{quad}
\end{equation}%
Note that $A>0$ by assumption. \ The positive root of equation
$\left( \ref{quad}\right)$ is 
\[
p^{\ast }=\frac{B}{k\left( k-A\right) }\left( \sqrt{1+\frac{A}{B}\left(
k-A\right) }-1\right) .
\]%
One can verify that $p^{\ast }\in \left[ 0,1\right] $, and conclude that $%
\left( p^{\ast },x^{\ast }\right) =\left( p^{\ast },\hat{x}\left( p^{\ast
}\right) \right) $ is the optimal point. 
\end{proof}

\begin{remark}
Case $A=0.$ By definition $\left( \ref{Adef}\right) $ we see that $%
1-m_{Y}=0.$ Note that $p^{\ast }=0$ follows from the identity above. Using
the definition and the computed values of $m_{F}=1-kp$ and $\sigma
_{F}^{2}=k^{2}p\left( 1-p\right) $ we write%
\begin{eqnarray*}
\hat{x}\left( p\right)  &=&\frac{m_{Y}-1+kp}{2d^{2}\left( \sigma
_{Y}^{2}+k^{2}p\left( 1-p\right) \right) }+\frac{k^{2}p\left( 1-p\right) }{%
\sigma _{Y}^{2}+k^{2}p\left( 1-p\right) } \\
\hat{x}\left( 0\right)  &=&\frac{m_{Y}-1}{2d^{2}\sigma _{Y}^{2}}=0.
\end{eqnarray*}%
In other words, when $A=0$ there is no risk and no expected loss in the
cryptocurrency. Thus, $\mathcal{D}$ realizes that any nonzero value of $p$
will result in $\mathcal{W}$ investing nothing in the home currency, F.

Case $A=k.$ The quadratic numerator $\left( \ref{quad}\right) $ is then $%
Q\left( p\right) =AB-2Bkp$ so that one has $p^{\ast }=1/2.$

Case $k<A\leq 2k$. By considering a small positive perturbation, $\delta ,$
of $A$ we see that $Q\left( \frac{1}{2}\right) >0$ so that the positive
region of $\partial _{p}U_{D}$ is extended toward the right as $A$ increases.

Case $A\geq 2k.$ Since $p\leq 1$ one has 
\[
Q\left( p\right) \geq B\left( A-2k\right) +k^{2}p^{2}\left( A-k\right) >0, 
\]%
so $\partial _{p}U_{D}>0$ and the maximum is thus $p^{\ast }=1.$
\end{remark}

\section*{Conclusion}

We have examined the optimal strategies for the key parties (those with
savings at risk, a dictatorial government and speculators) involved explicitly
or implicitly in the formation of an equilibrium for cryptocurrencies.\ The
second method involves different time scales in determining equilbrium that
differs from the more common Nash equilibrium, in which all parties can
readjust their positions continuously. As described in Section 4 this can be
utilized for many realistic situations in which one entity such as a
government optimizes by placing conditions such as taxes, tariffs, fees, etc.,
or policies that cannot be reversed or adjusted in a short time. In general,
optimization in this form favors the group that can make immediate and
continuous adjustments.

Each of the methods are based on parameters that can be estimated. For
example, the variance of cryptocurrencies can be determined from the trading
data. Parameters such as $k$ (the fraction of assets seized) can be estimated
from the policies of the government. An assessment of these quantities then
leads to estimates of the amount of money that is likely to be used to
purchase cryptocurrencies in the aggregate. Using the ideas summarized in
\cite{cc} one can then also evaluate average price changes of cryptocurrencies
as well as the total market capitalization of cryptocurrencies. The evolution
of the latter is crucial in understanding the implications of instability of
cryptocurrencies on other sectors of the world's economy.

Major governments have often appeared confused and lethargic in their response
to cryptocurrency policy, even insofar as deciding whether it is important or
not. There is also little understanding of the conditions under which a
cryptocurrency could be either beneficial or detrimental to global society.
The perspective of our paper suggests that a cryptocurrency price will vary
widely depending on the demand that in turn is based on policies of countries
where monetary policy and laws, in general are less developed. Together with
the fact that cryptocurrencies cannot be redeemed for any asset, one cannot
expect much stability. However, given a mechanism whereby a cryptocurrency is
essentially backed by real assets (e.g., a structure similar to Exchange
Traded Funds) one would have stability since arbitrageurs would take advantage
of any discrepanicies. This could be linked of course to a single currency
such as the US\ Dollar, but would only be a trading token in this case.

However, one can design a cryptocurrency that would essentially grow with the
world's economy, unlike a commodity such as gold. A simple example would be
that the cryptocurrency could be reedemable in units of the Gross World
Product in terms of a basket of major currencies, so that each cryptocurrency
could be redeemed for one trillionth of the GWP in Dollars, Euros and Yen.
Such an instrument would offer much greater stability and could be used as a
substitute currency that is independent of any government. As shown in our
analysis, as the volatility risk would diminish, and those whose assets in the
home currency are at risk would place more of their assets into this
cryptocurrency. Thus the fraction, $x,$ placed in the cryptocurrency would
increase. In particular, the equilibrium point $\left(  p^{\ast},x^{\ast
}\right)  $ would feature $x^{\ast}$ that is larger and $p^{\ast}$ that is
smaller. This would mean that the citizens have greater economic freedom, and
financially totalitarian regimes would have smaller resources. In summary, the
creation of a viable cryptocurrency with intrinsic value would have less
volatility, and thereby reduce the fraction of savings in the home currency
that is under threat by a totalitarian government, whose existence is often
contingent on raising money in this manner.

\section{Acknowledgements}

The authors thank the Economic Science Institute at 
Chapman University and the Hayek Foundation for their support. 
Discussions with Prof. Gabriele Camera are very much appreciated.

\section*{Appendix A: Fundamental Value and Liquidity Value}

There is a temptation to
stipulate that the \textit{only} valuation of an asset is the trading price,
as this price reflects the preferences and values of the buyers and sellers
via the intersection of the supply and demand curves. In principle there is
nothing wrong with this perspective except that important phenomena are left
unexplained, and significant risks are mischaracterized as rare or low risk.

One way to examine different aspects of price or value is through the
laboratory experiments such as the "bubbles"\ experiments introduced by
\cite{SSW} in which an asset pays a dividend with expectation 24 cents at the
end of each of 15 periods, and is then worthless. The value of this asset at
the end of Period $k$ is clearly given in dollars by
\begin{equation}
P_{a}:=3.60-\left(  0.24\right)  k. \label{a1}%
\end{equation}
In numerous experiments, prices often started well below $\left(
\ref{a1}\right)  $ and soared far above this fundamental value, and eventually
crashed. This persisted even in experiments in which the dividend payout had
no randomness at all \cite{PS}. It seems difficult to deny that $P_{a}$ is a
meaningful and useful quantity, particularly since it is a quantity that the
trader can be assured of receiving. For example, purchasing at a trading price
early in the experiment that is often below $P_{a}$ ensures the trader will
gain a specific profit. If one ignores the intrinsic value, $P_{a}$, one would
conclude that the risk is the same at any price, and likely incur a large loss
during the course of the experiment. Also, it has been noted \cite{cb}, that
in these experiments, there is a third quantity with units of price per share.
This is the "liquidity value," $L,$ defined as the ratio of the sum of all
cash in the experiment divided by the total of all shares. Experiments that
were designed to test this concept \cite{CPS1, CPS2} showed that the liquidity
value has a primary role in determining the size of the bubble. In fact the
peak of the bubble was often close to $L.$ In other words, when traders pay
little attention to fundamental value, or if the fundamental value is not
clear, the price drifts toward the liquidity value \cite{cb}. At the opposite
extreme, for short term\ government bonds, the calculation of fundamental
value is clear, as the owner is assured of a particular sum at a particular
time a few months in the future. The trading price generally trades very close
to this fundamental value since there are many arbitrageurs who exploit any deviations.

The vast majority of cryptocurrencies do not have any redemption value, they
pay no dividends, and they do not endow holders with voting power over an
entity with assets (as do stocks, for example). Thus, classical finance
calculations involving expected dividends, book value, replacement value,
etc., all yield a fundamental value of zero. One exception is JP Morgan's JPM
Coin, announced in February 2019 which would be redeemable in US dollars. The
redemption price would yield the guaranteed value, which would be $P_{a}$, the
fundamental value, so long as the investors are confident in JP Morgan's
ability to fulfil its commitment.

Ignoring fundamental or intrinsic value often leads to disastrous practical
results, as investors discovered with the internet stocks in 1999, or the
Japanese market in 1990, for example, when standard calculations of stock
value \cite{g, gm} showed a large discrepancy between the trading price and
the fundamental value.

Similarly, in theoretical development, neglecting either the fundamental
value, $P_{a}$, or the liquidity value, $L,$ will have the same consequences
as overlooking any other important quantity in modeling economics problems.
One obtains some results that are not consistent with observations, and has no
way to rectify the situation.

Although one cannot calculate a positive $P_{a}$ for the typical
cryptocurrency, people are paying for these units, so that they must see some
value in it. The perspective that fundamental value must be the trading price,
renders the equivalence a tautology. As discussed above, the result is that
important phenomena are left unexplained, and an even a basic understanding of
the likely price evolution becomes more difficult.

Since cryptocurrencies have no fundamental value, prices will naturally drift
toward the liquidity value, which will be given by the total amount of cash
available for the cryptocurrency (i.e., demand) divided by the number of units
\cite{CPS1, CPS2}.

The absence of a non-zero fundamental value means that price will be set by
the supply (which is fixed, for example, for Bitcoin) \ and demand in
accordance with equation $\left(  \ref{demand}\right)  $. Thus it is a
calculation of demand that is key to understanding equilibrium price.

\section*{Appendix B}

1. We consider first the role of "pure" speculators who have no control of the
type of trading or auction, the rules of the exchange, the enforcement of the
rules, the display of orders, and the flow of information. Volatility arises
endogenously due to the various trading strategies, such as trend following,
and random events that motivate any of the traders. For many first time or
novice traders, there is a tendency to overreact, and to chase a trend, or hop
onto a fad. In the case of cryptocurrencies, which lack any fundamental value,
any news is likely to result in an overreaction. Thus volatility can be
expected to be high in the absence of any anchor. For example, Treasury bills
offer a guaranteed payout, so that a small deviation from the certain payout
due within a few months would be exploited by arbitrageurs and the price would
be restored close to its intrinsic value. The speculators in many markets have
a better understanding (compared to novice traders) of the factors that move
prices within a short time scale. Speculators are generally believed to lower
volatility \cite{BBH}, as they use their capital to buy when prices move
unjustifyably lower. Of course, when prices exhibit very low volatility, there
is no financial incentive for speculators to trade. Consequently, in an
idealized setting, the short-term volatility level will be established as the
minimum value at which speculators find adequate profits after costs.

\bigskip

2. Next we consider "speculators"\ in less established markets in which the
rule makers, market makers, news makers are all essentially the same group. In
most developed markets such as the New York Stock Exchange (NYSE) and major
commodity exchanges there are precise rules designed to promote fairness and
ease of trading that have been developed over many years. An example is the
NYSE rule that if there are two orders to purchase a stock, it is the higher
one that prevails. Surprisingly to novice traders, this is not usually a
feature of most markets. In many markets there are "market makers"\ who are
entitled to buy the stock for their own account at a lower price, even though
a higher bid has been placed by a retail customer or trader. The rules of each
exchange endow the market makers and market specialists with the power to buy
and sell on their own account. In many well-developed exchanges, there are
rules agains "front-running" whereby insiders buy on their own accounts as
they become aware of a set of large orders that are entering the market.
Another example on major exchanges involves "not held" trades that are placed
with the market makers but are not displayed. The intention here is that a
large order to sell could prompt further selling by less informed traders. By
contrast, in a less developed market environment, a market insider can place a
large order (but above the market price) that will immediately lead to lower
prices, whereupon he can deftly purchase.

Novice traders usually make numerous assumptions relating to fairness on the
nature of market rules and procedures. Unfortunately, these are generally
false for less developed markets that cater to inexperienced traders. The
wishful thinking of new traders seeking quick riches (or escape from a
currency) provides for a healthy income for those dominating these markets in
terms of making the rules (if there are any at all) and using their capital to
control the volatility. For many of the cryptocurrencies, for example, it is
not even clear what the rules are, or where they would be enforced. Thus, in
an under-developed market, a group of participants that controls the rules of
trading has numerous tools at its disposal to adjust volatility. Even the
hours of trading have a strong impact on volatility. For example, it is
well-known that trading around the clock leads to times periods of low volume
so that a few trades can move prices much more than during actively traded
times. On the other hand, in an exchange in which there is a single trade each
day at a specified time, the maximum minus minimum price within one week is
likely to be much lower than in 24 hour trading.

Another feature that can influence prices is the extent to which information
on orders is displayed. The "order book" displays the array of bids and asks
for the asset in continuous time. Whether or not the order book is displayed
depends upon the rules of the exchange. Also, on some exchanges, the market
maker can choose to display only some of the orders. In laboratory experiments
\cite{CPS2} it was shown that bubbles are tempered by the display of the
complete order book.

Related to the order book are the rules under which the market maker can buy
for his own account. While "front running" -- the practice of buying for one's
own account ahead of a large order -- or "shadowing" -- buying the same assets
as a particular trader -- are banned in some of the most developed exchanges,
one cannot assume that they will be prohibited universally.

Of course, all of this assumes that there is a real market in which bids and
asks are matched with some rule. In many cases purchases and sales are made
through one entity that buys and sells for its own account, thereby granting
itself a generous profit as the middleman. Even in large brokerages it is
common for monthly statements to disclose "we make a market in this stock"
that indicates the bid/ask spread is whatever the company designates as
revenue for itself. From the perspective of the individual trader, the bid/ask
spread, of course, adds to the cost and volatility of the transaction.

\section*{Appendix C}

In examining the choice faced by $\mathcal{W}$ we assume
that one option is to remain in the home currency, F, and the other to buy
the cryptocurrency with the objective of later selling in order to buy other
assets such as a more reliable currency, gold, etc.

The group $\mathcal{W}$ experiences a loss or gain on these transactions with the
speculators, group $\mathcal{S}$, which itself has a non-linear utility function
reflecting that fact that high volatility is good for profits up to a point
after which it has a negative impact. Thus, one has the following utility
function for group $\mathcal{S}$:%
\[
U_{S}=a_{1}V-a_{2}V^{2}
\]%
where $V$ represents the volatility or variance, for example, and $a_{1}$, $%
a_{2}$ are positive constants. Hence, there will be a maximum value $V=V_{m}$
that maximizes the utility of the speculators. This can be viewed as a fixed
quantity from the perspective of $\mathcal{W}$.

The mean, $m_{Y}$, and variance, $\sigma _{Y}^{2}$ of group $\mathcal{W}$'s
investments in Y can be calculated based on $V_{m}$ and the other parameters
that describe the trading. In particular, we assume that there is a
probability $q$ (presumably small) that $\mathcal{W}$ will profit, and that
their wealth will increase from $1$ to $1+r_{1}V_{m}$ \ and a probability $%
1-q$ that it will decrease from $1$ to $1-r_{2}V_{m}$ \ where $r_{2}>r_{1}>0.
$ In other words, there is a small probabilty, $q$, that $\mathcal{W}$ will benefit by $%
r_{1}V_{m}$ (as a fraction of their original wealth) and a larger
probability, $1-q$, that they will lose a larger sum $r_{2}V_{m}$ . The loss
is proportional to the volatility \ as the professional speculators are able
to exploit the ups and downs of the trading at the expense of the
inexperienced $\mathcal{W}$.

The mean and variance of the outcome are then%
\begin{eqnarray*}
m_{Y} &=&q\left( 1+r_{1}V_{m}\right) +\left( 1-q\right) \left(
1-r_{2}V_{m}\right) , \\
\sigma _{Y}^{2} &=&q\left( 1-q\right) \left( r_{1}+r_{2}\right) V_{m}.
\end{eqnarray*}%
In other words, there is large probability that $\mathcal{W}$ will take a loss on the
transaction. One can consider more general probability distributions for 
$\mathcal{W}$'s profits and losses, but ultimately, the two quantities that
are relevant for its utility function $U_{W}$ are given by $m_{Y}$ and $
\sigma _{Y}^{2}$ that one can regard as empirical observables.

\end{document}